\documentclass[singlecolumn]{article}

\usepackage{graphicx}

\usepackage{amssymb}
\usepackage{amsmath}
\usepackage{epstopdf}
\usepackage{caption}
\usepackage{subcaption}
\usepackage[sort&compress,numbers]{natbib} 
\usepackage{dcolumn}
\usepackage{bm}
\usepackage{mathtools}

\DeclarePairedDelimiter\floor{\lfloor}{\rfloor}
\usepackage{amsthm}
\theoremstyle{plain}
\newtheorem{thm}{Theorem}
\newtheorem{lem}[thm]{Lemma}

\usepackage{color}
\usepackage{hyperref}

\title{Cost Effective Campaigning in Social Networks}
\author{Bhushan Kotnis\footnote{Indian Institute of Science, bkotnis@dese.iisc.ernet.in} and Joy Kuri\footnote{Indian Institute of Science}}

\begin{document}
\maketitle




\begin{abstract}
Campaigners are increasingly using online social networking platforms for promoting products, ideas and information. A popular method of promoting a product or even an idea is incentivizing individuals to evangelize the idea vigorously by providing them with referral rewards in the form of discounts, cash backs, or social recognition. Due to budget constraints on scarce resources such as money and manpower, it may not be possible to provide incentives for the entire population, and hence incentives need to be allocated judiciously to appropriate individuals for ensuring the highest possible outreach size. We aim to do the same by formulating and solving an optimization problem using percolation theory. In particular, we compute the set of individuals that are provided incentives for minimizing the expected cost while ensuring a given outreach size. We also solve the problem of computing the set of individuals to be incentivized for maximizing the outreach size for given cost budget. The optimization problem turns out to be non trivial; it involves quantities that need to be computed by numerically solving a fixed point equation.  Our primary contribution is, that for a fairly general cost structure, we show that the optimization problems can be solved by solving a simple linear program. We believe that our approach of using percolation theory to formulate an optimization problem is the first of its kind.
\end{abstract}


\section{Introduction \label{sec:Introduction}}
Online social networking platforms are being increasingly used by campaigners, activists and marketing managers for promoting ideas, brands and products. In particular, the ability to recommend news articles  \cite{Leskovec2009}, videos, and even products \cite{Leskovec2006} by friends and acquaintances through online social networking platforms is being increasingly recognized by marketing gurus as well as political campaigners and activists. Influencing the spread of content through social media enables campaigners to mold the opinions of a large group of individuals. In most cases, campaigners and advertisers aim to spread their message to as many individuals as possible while respecting budget constraints. This calls for a judicious allocation of limited resources, like money and manpower, for ensuring highest possible outreach, i.e., the proportion of individuals who receive the message.    
\par
Individuals share information with other individuals in their social network using Twitter tweets, Facebook posts or simply face to face meetings. These individuals may in turn pass the same to their friends and so on, leading to an information epidemic. However, individuals may also become bored or disillusioned with the message over time and decide to stop spreading it. { Past research suggests that such social effects may lead to opinion polarization in social systems \cite{Sinha2006}}. This can be exploited by a campaigner who desires to influence such spreading or opinion formation by incentivizing individuals to evangelize more vigorously by providing them with referral rewards in the form of discounts, cash back or other attractive offers. Due to budget constraints, it may not be feasible to incentivize all, or even a majority of the population. Individuals have varying amount of influence over others, e.g., ordinary individuals may have social connections extending to only close family and friends, while others may have a large number of social connections which can enable them to influence large groups \cite{Goldenberg2009}. Thus, it would seem that incentivizing highly influential individuals would be the obvious strategy. However, recruiting influential people can be very costly, which may result in the campaigner running out of funds after recruiting just a handful of celebrities, which in turn may  result in suboptimal outreach size. 
\par
 A resource constrained campaigner, for a given cost budget, may want to maximize the proportion of informed individuals, while other campaigners who care more about campaign outreach than resource costs, may desire to minimize costs for achieving a given number of informed individuals. We  address both the resource allocation challenges by formulating and solving two optimization problems with the help of \emph{bond percolation theory}. 
 \par
A similar problem of preventing epidemics through vaccinations has received a lot of attention \cite{Cohen2003,Shaw2010,Ruan2012,Starnini2013,Peng2013}. However, in these problems the cost of vaccination is uniform for all individuals, and hence it is sufficient to calculate the minimum number of vaccinations. Information diffusion can also be maximized by selecting an optimal set of seeds, i.e., individuals best suited to \emph{start} an epidemic \cite{Kempe2003,Chen2009a,Chen2010a}. This is different from our strategy which involves incentivizing individuals to \emph{spread} the message.  It is possible to address the problem posed here using optimal control theory, which involves computing the optimal resource allocation in real time for ensuring maximum possible outreach size by a give deadline \cite{Karnik2012,Dayama2012, Kandhway2014a,Kandhway2014,Kandhway2014b}.  However, the optimal control solution is not only difficult to compute, but also very hard to implement as it requires a centralized real time controller. Furthermore, recent work, \cite{Karnik2012,Dayama2012, Kandhway2014a,Kandhway2014,Kandhway2014b}, on optimal campaigning in social networks does not address the problem of minimizing the cost while gurantering an outreach size. Our formulation allows us to solve both the problems.
\par
Our model assumes two types of individuals viz. the \emph{`ordinary'} and the \emph{`selected'}, and they are connected to one another through a social network.   Before the campaign starts, the selected individuals are incentivized to spread the message more vigorously than the ordinary. We use the \emph{Susceptible Infected Recovered} (SIR) model for modeling the information epidemic.  For a given set of selected individuals, we first calculate the size of the information outbreak using network percolation theory, and then find the set of selected nodes which, 1. minimizes the cost for achieving a given proportion of informed individuals, and 2. maximize the fraction of informed individual for a given cost budget. We believe that our approach of using percolation theory to formulate an optimization problem is the first of its kind.
\par
The detailed model description can be found in Sec. \ref{sec:Model},  percolation analysis in Sec. \ref{sec:Analysis}, the problem formulation in Sec. \ref{sec:Problem Formulation}, numerical results in Sec. \ref{sec:Numerical Results}, and finally conclusions are discussed in Sec.  \ref{sec:Conclusion}.
\section{Model \label{sec:Model}}
We divide the total population of $N$ individuals in two types: the ordinary (type $1$) and the selected (type $2$). Before the campaign starts selected individuals are provided incentives to spread the information more vigorously. These individuals are connected with one another through a social network, which is represented by an undirected graph (network). Nodes represent individuals while a link embodies the communication pathways between individuals. 

Let $P(k)$ be the degree distribution of the social network. For analytical tractability, we assume that the network is uncorrelated \cite{Barrat2008}. We generate  an uncorrelated network using the configuration model \cite{molloy1995}. A sequence of $N$ integers, called the degree sequence, is obtained by sampling the degree distribution. Thus each node is associated with an integer which is assumed to be the number of half edges or stubs associated with the node. Assuming that the total number of stubs is even, each stub is chosen at random and joined with another randomly selected stub. The process continues until all stubs are exhausted. Self loops and multiple edges are possible, but the number of such self loops and multiple edges goes to zero as $N \to \infty$ with high probability. We assume that $N$ is large but finite.  Let $\phi (k)$ be the proportion of individuals with $k$ degrees that are provided incentives for vigorously spreading the message, i.e., proportion of nodes with degree $k$ that are type 2 nodes. The goal is to find the optimum $\phi(k)$ for maximizing the epidemic size (or minimizing the cost). The actual individuals can be identified by sampling from a population of individuals with degree $k$ with probability $\phi(k)$.
\par
We assume that the information campaign starts with a randomly chosen individual, who may pass the information to her neighbors, who in turn may pass the same to their neighbors and so on. However, as the initial enthusiasm wanes, individuals may start loosing interest in spreading the information message. This is similar to the diffusion of infectious diseases in a population of susceptible individuals. Since, we account for individuals loosing interest in spreading the message, we use a continuous time SIR process to model the information diffusion. The entire population can be divided into three classes, those who haven't heard the message (susceptible class), those who have heard it and are actively spreading it (infected class) and those who have heard the message but have stopped spreading it (recovered class). 
\par
Let $\beta_1$ be the rate of information spread for an ordinary node (Type 1), while $\beta_2$ for a selected node (Type 2). In other words, the probability that a type $i$ individual `infects' her susceptible neighbors in small time $dt$ is $\beta_idt +o(dt)$. Note that this is independent  of the type of the susceptible node. Let $\mu_i$ be the rate at which type $i$ infected individuals move to the recovered state. The larger the $\mu_i$ the lesser the time an individual spends in spreading the message. Since type $2$ individuals are incentivized to spread information more vigorously, $\beta_2 > \beta_1$ and $\mu_2 < \mu_1$. Let $T_i$ be the probability that a type $i$ infected node infects its susceptible neighbors (any type) before it recovers ($i \in \{0,1\}$). It can be easily shown that $T_i = \frac{\beta_i}{\beta_i + \mu_i}$, see \cite{Newman2002}. Therefore, $T_2 > T_1$. $T_i$ can be interpreted as the probability that a link connecting type $i$ infected node to any susceptible node is occupied. We refer to such links as type $i$ links and $T_i$ the occupation probability for link of type $i$. This mapping allows us to apply bond percolation theory for obtaining the size of the information epidemic \cite{Newman2010}. 

 \section{Analysis \label{sec:Analysis}}
We first aim to calculate the proportion of individuals who have received the message, or in other words, the proportion of recovered individuals at $t \to \infty$. Let $P(k' \mid k)$ be the probability of encountering a node of degree $k'$ by traversing a randomly chosen link from a node of degree $k$. In other words, $P(k' \mid k)$ is the probability that a node with degree $k$ has a neighbor with degree $k'$. For a network generated by configuration model, $P(k' \mid k)= \frac{k'P(k')}{\langle k \rangle}$ \cite{Newman2010}, where $\langle k^i \rangle$ is the $i^{th}$ moment of $P(k)$.
 \par
  Let $q$ be the probability of encountering a type 2 node by traversing a randomly chosen link from a node of degree $k$. Therefore, $q = \sum\limits_{k'=1}^{\infty}Pr($Neighboring node is type 2 $\mid$ neighboring node has degree $k')\cdot Pr($Neighboring node has degree $k'\mid$ original node has degree $k)$. 
   \begin{align*}
   q= \frac{1}{\langle k \rangle}\sum_{k=1}^{\infty}k\phi(k)P(k)
   \end{align*}
   The probability that a randomly chosen node has $k_1$ type 1 and $k_2$ type 2 neighbors $= P(k_1,k_2) = \sum\limits_{k:k=k_1+k_2}^{\infty} Pr(k_1,k_2\mid$node has degree $k)P(k) $. 
    \par
    For a large $N$, the event that a given node has degree $k$, can be approximated to be independent of the event that another node, having a common neighbor with the given node, has degree $k'$. This is true since the degree sequence is generated by independent samples from the distribution, and for a large $N$ the effect of sampling without replacement is negligible. The probability that a node is selected (type 2), is a function of its degree, hence the event that a node is type 1 (or 2) is independent of the event that any other node is type 1 (or 2). This allows us to write:   
    \begin{align*}
     P(k_1,k_2)=  {k_1+k_2 \choose k_2}q^{k_2}(1-q)^{k_1}P(k_1+k_2) 
    \end{align*}
Let $Q(k)$ be the excess degree distribution, i.e., the degree distribution of a node arrived at by following a randomly chosen link without counting that link. For the configuration model $Q(k) = (k+1)P(k+1)/ <k>$. Let $Q(k_1,k_2)$ be the excess degree distribution for connections to type 1 and type 2 nodes.
  \begin{align*}
  Q(k_1,k_2) = {k_1+k_2 \choose k_2}q^{k_2}(1-q)^{k_1}Q(k_1+k_2)
  \end{align*}
  Let $\tilde{P}(\tilde{k}_1,\tilde{k}_2)$ and $\tilde{Q}(\tilde{k}_1,\tilde{k}_2)$ be the distribution and the excess distribution of the number of type 1 and type 2 neighbors that have received the information message. In other words the distribution and the excess distribution of type $i$ occupied links.
   {\small
   \begin{align*}
   \tilde{P}(\tilde{k}_1,\tilde{k}_2) &= \sum_{k_1=\tilde{k}_1}^{\infty}\sum_{k_2=\tilde{k}_2}^{\infty}P(k_1,k_2)\prod_{i=1}^{2}{k_i \choose \tilde{k}_i}T_i^{\tilde{k}_i}(1-T_i)^{k_i-\tilde{k}_i} \\
   \tilde{Q}(\tilde{k}_1,\tilde{k}_2) &= \sum_{k_1=\tilde{k}_1}^{\infty}\sum_{k_2=\tilde{k}_2}^{\infty}Q(k_1,k_2)\prod_{i=1}^{2}{k_i \choose \tilde{k}_i}T_i^{\tilde{k}_i}(1-T_i)^{k_i-\tilde{k}_i}
   \end{align*}
   }
   \begin{table}[!t]
         \centering
           \begin{tabular}{l  l}
      \hline
      Generating function & Distribution \\
       \hline
      $G(u_1,u_2)$ & $P(k_1,k_2)$ \\
       $F(u_1,u_2)$ & $Q(k_1,k_2)$ \\
          $\tilde{G}(u_1,u_2)$ & $\tilde{P}(\tilde{k}_1,\tilde{k}_2)$ \\
           $\tilde{F}(u_1,u_2)$ & $\tilde{Q}(\tilde{k}_1,\tilde{k}_2)$ \\
      $\tilde{H}_i(u_1,u_2)$ & Proportion of type 1 and type 2 nodes, \\ 
      & who have received the message, in a component \\   
      & reached from a type $i$ link.\\
          $\tilde{J}_i(u_1,u_2)$ &  No. of type 1 and type 2 nodes \\  
         & who have received the message, in a component \\ 
            & reached from a node $i$.\\
             $\tilde{J}(u_1,u_2)$ &  No. of type 1 and type 2 nodes \\  
      & who have received the message, in a component \\ 
         & reached from a randomly chosen node.\\
        \hline
        \end{tabular}
        \caption{List of probability generating functions.}
        \label{table:pgf}
              \end{table}    
               The probability generating functions  for the distributions used in the analysis above are listed in Table \ref{table:pgf}.  For example $G(u_1,u_2)$ is given by :
                \begin{align*}
                G(u_1,u_2) = \sum\limits_{k_1,k_2=0}^{\infty} u_1^{k_1}u_2^{k_2}P(k_1,k_2)
                \end{align*} 
               Now, $\tilde{G}(u_1,u_2)$ is given by
                \begin{align*}
                & \sum_{\tilde{k}_1,\tilde{k}_2}^{\infty}u_1^{\tilde{k}_1}u_2^{\tilde{k}_2}\sum_{k_1 = \tilde{k}_1}\sum_{k_2=\tilde{k}_2}P(k_1,k_2)\prod_{i=1}^{2}{k_i \choose \tilde{k}_i}T_i^{\tilde{k}_i}(1-T_i)^{k_i-\tilde{k}_i} \\
                &= \sum_{k_1,k_2}^{\infty}(1+(u_1-1)T_1)^{k_1}(1+(u_2-1)T_2)^{k_2}P(k_1,k_2) \\
                &= G\left(1+(u_1-1)T_1,1+(u_2-1)T_2\right)
               \end{align*}
               Similarly, $\tilde{F}(u_1,u_2) = F(1+(u_1-1)T_1,1+(u_2-1)T_2)$ 
               \par
 A component is a \emph{small} cluster of nodes that have received the information message. By small we mean that the cluster is finite and does not scale with the network size. However, at the phase transition, the average size of the cluster diverges (as $N \to \infty$). An information epidemic outbreak is possible only when the average size of the cluster diverges. In this regime the component is termed as a giant connected component (GCC) and it grows with the network size. Let $\tilde{H}_i(u_1,u_2)$ be the generating function of the distribution of the number of type 1 and type 2 nodes in a component arrived at from a type $i$ link. Let  $\tilde{J}_i(u_1,u_2)$ and $\tilde{J}(u_1,u_2)$ be the generating functions of the distribution of the number of type 1 and type 2 nodes in a component arrived at from node $i$ and a randomly chosen node, respectively. 
\par
 \begin{figure}
 \centering
  \includegraphics[width = 0.6\textwidth]{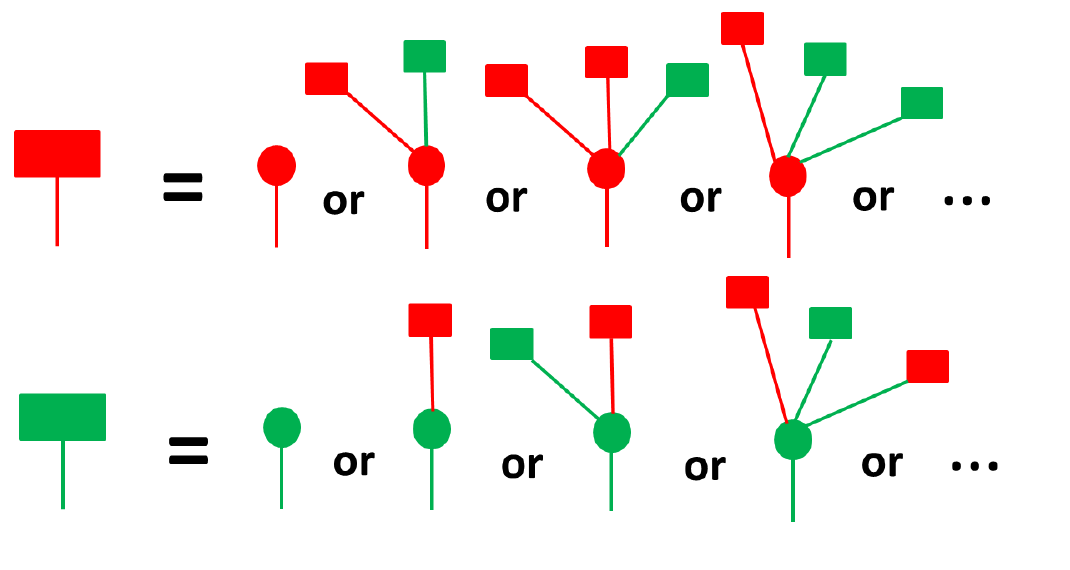}
 \caption{(Color Online) Illustration of components. The red boxes represent the components reached by type $1$ link while green boxes represent components reached by a type $2$ link. A type 2 node is represented by a green circle while a red circle represents a type 1 node. }
 \label{fig:percolation}
 \end{figure}
 Let random variable $Y_i$ be the number of type 1 and 2 nodes, that have received the message, in a component arrived at from a type $i$ link.
  The probability of encountering closed loops in finite cluster is $O(N^{-1})$ \cite{Newman2002} which can be neglected for large $N$. The tree like structure of the cluster allows us to write the size of the component encountered by traversing the link, as the sum of the size of components encountered after traversing the links emanating from the node at the end of the initial link. This is illustrated in Fig. \ref{fig:percolation}.  Hence, $Y_i$ can be written as:
 \begin{align*}
 Y_i = 1 + \tilde{K}_1Y_1 + \tilde{K}_2Y_2
 \end{align*}
 where random variable $\tilde{K}_i$ is the number of  type $i$ neighbors of the end node of type $i$ link that have received the message; the arrival link is not counted (excess degree). Since, the size of the components along different links are mutually independent (absence of loops) we can write the above equation in terms of probability generating functions. 
 \begin{align*}
 \tilde{H}_i(u_1,u_2) &= u_i\tilde{H}_1(u_1,u_2)^{\tilde{K}_1}\tilde{H}_2(u_1,u_2)^{\tilde{K}_2} \\ 
 &= u_i \sum_{\tilde{k}_1,\tilde{k}_2}^{\infty} \tilde{H}_1^{k_1}(u_1,u_2)\tilde{H}_2^{k_2}(u_1,u_2)\tilde{Q}(\tilde{k}_1,\tilde{k}_2) \\
 &= u_i\tilde{F}(\tilde{H}_1(u_1,u_2) ,\tilde{H}_2(u_1,u_2) ) \\
 \end{align*}
 Which can also be written as 
  \begin{align}
\tilde{H}_i(u_1,u_2) = u_iF\left(1+(\tilde{H}_1(u_1,u_2)-1)T_1,1+(\tilde{H}_2(u_1,u_2)-1)T_2 \right) \label{eqn:pgfCluster}
 \end{align}
  Similarly, $\tilde{J}(u_1,u_2)$ can be expressed as :
 \begin{align*}
 \tilde{J}_i(u_1,u_2) &= u_i\sum_{\tilde{k}_1,\tilde{k}_2}^{\infty} \tilde{H}_1^{k_1}(u_1,u_2)\tilde{H}_2^{k_2}(u_1,u_2)\tilde{P}(\tilde{k}_1,\tilde{k}_2) \\
 &= u_i\tilde{G}(\tilde{H}_1(u_1,u_2) ,\tilde{H}_2(u_1,u_2) ) \\
 \tilde{J}(u_1,u_2) &= (1-p)\tilde{J}_1(u_1,u_2) +p \tilde{J}_2(u_1,u_2)
 \end{align*}
 where $p$ is the probability of choosing a type 2 node, $p = \sum \limits_{k=1}^{\infty}P(k)\phi(k)$.
 The following theorem  describes the phase transition conditions required for an outbreak and the size of the such an outbreak.  The proof can be found in  \ref{appendix:theorem1}.
 
 \begin{thm}
 The condition required for a small cluster to become a giant connected component is given by: $\tilde{\nu}\geq 1 $, where 
 \begin{align*}
 \tilde{\nu} = T_1\sum_{k_1,k_2}^{\infty}k_1Q(k_1,k_2) + T_2\sum\limits_{k_1,k_2}^{\infty}k_2Q(k_1,k_2) 
 \end{align*}
  and the proportion of nodes in the giant connected component (size of GCC) is given by $1-\psi$,
 \begin{align*}
 \psi = \sum_{k_1,k_2}^{\infty}(1+(u^*-1)T_1)^{k_1}(1+(u^*-1)T_2)^{k_2}P(k_1,k_2)
 \end{align*}
 where $u^*$ is the solution of the fixed point equation
 \begin{align*}
 u =  \sum_{k_1,k_2}^{\infty}(1+(u-1)T_1)^{k_1}(1+(u-1)T_2)^{k_2}Q(k_1,k_2)
 \end{align*}
 \end{thm}
 \par
 The size of the information epidemic outbreak can now be used for formulating the optimization problem.
 \section{Problem Formulation \label{sec:Problem Formulation}}
 Providing incentives in the form of referral rewards for low degree nodes, or sponsorship offers for celebrities (high degree nodes) is costly. Since, the cost is a function of the degree let  $c(k)$ be the cost of providing incentivizing a node with degree $k$. The average cost, $\bar{c}(\boldsymbol{\phi})$, is given by $\sum\limits_{k=1}^{\infty}c(k)Pr($node is selected  $\mid$ node has degree $k)P(k) = \sum\limits_{k=1}^{\infty} c(k)\phi(k)P(k)$.  The proportion of type 2 individuals is given by $\sum\limits_{k=1}^{\infty}\phi(k)P(k)$.
  \par
 We formulate two optimization problems, viz., one which minimizes cost while enforcing a lower bound on the epidemic size, and the other which maximizes the epidemic size for a given cost budget. For both the problems, the evaluation of the size of the epidemic requires one to numerically solve a fixed point equation. Thus, there is no straightforward method to solve the optimization problem such as the Karush Kuhn Tucker (KKT) conditions, because evaluating the objective function requires one to solve a fixed point equation. We show that this problem can be reduced to a linear program, which can then be solved easily using any off the shelf LP solver.  
 \subsection{Cost minimization problem}
 Providing guarantees on the minimum number of individuals who will be informed about the campaign is appropriate for campaigns with large funding, such as election campaigns where message penetration is more important than the cost. The guarantee on epidemic size is written as a constraint to the optimization problem. The cost $\bar{c}(\boldsymbol{\phi})$ is minimized subject to $1 - \psi \geq \gamma$ where $\gamma \ \in \ [0,1]$ and $\boldsymbol{\phi}$ is the control variable.  If $\gamma = 0$, the constraint becomes $\tilde{\nu} \leq 1$, as $\gamma = 0$ implies $\psi = 1$ which is the same as $\tilde{\nu} \leq 1$. A finite amount of money, may put a constraint on the number of type 2 individuals. The proportion of type 2 individuals is given by $\sum_{k=1}^{\infty} \phi(k)P(k)$. This translates in to the constraint : $\sum_{k=1}^{\infty} \phi(k)P(k) \leq B$, where budget $B \ \in \ [0,1]$. 
 \par
     The following theorem which is our principle contribution allows us to solve a possible non convex problem by solving a linear program.  The key insight is that the probability of outbreak is monotonically decreasing in $q$, which then allows one to write the optimization problem as a linear program.  The intuition behind this claim is that since $q$ is the probability of finding a type 2 node on a randomly chosen link, increase in $q$ is equivalent to the increase in number of type 2 individuals resulting in a higher epidemic size. 
   \begin{thm} \label{theorem2}
   If $T_2 > T_1$, then $\psi \ \in \  (0,1)$, is strictly decreasing with respect to $q$, i.e, $\frac{d\psi}{dq} <0$ for all $q \ \in \  [0,1] $. For the $\psi =0$ case $(\tilde{\nu}\geq 1)$, $\tilde{\nu}$ is strictly increasing with respect to $q$, i.e, $\frac{d\tilde{\nu}}{dq} > 0, \ \forall \ \ q \ \in \  [0,1]$, where $q = \frac{1}{\langle k \rangle}\sum\limits_{k=1}^{\infty}k\phi(k)P(k)$.
   \end{thm}
   \begin{proof}
   The proof follows from Lemmas  \ref{lemma2} and \ref{lemma3} detailed in Appendix \ref{appendix:theorem2}.
   \end{proof}
     Since, $\frac{d\psi}{dq} < 0$, the epidemic size constraint can be written as $\frac{1}{\langle k \rangle}\sum\limits_{k=1}^{\infty}k\phi(k)P(k) \geq q^*$, where $ \psi(q) \mid_{q=q^*} \ = 1-\gamma$. 
   The optimization problem can now be written as follows:
   \begin{align}
      \begin{aligned}
      & \underset{\boldsymbol{\phi}}{\text{minimize}}
       \ \ \ \ \sum_{k=1}^{\infty} c(k)\phi(k)P(k)  \\
      & \text{subject to:} \ \  \\
      &		  \frac{1}{\langle k\rangle}\sum_{k=1}^{\infty}k\phi(k)P(k) \geq q^* \\
      &			\sum_{k=1}^{\infty} \phi(k)P(k) \leq B      \\
      &		\boldsymbol{0} \leq \boldsymbol{\phi} \leq \boldsymbol{1}  \label{eqn:LinOpt1}
        \end{aligned}
   \end{align}
   The above problem is a linear program which can be solved by any off-the-shelf LP solver.  
   \par
   The optimization problem described above may not be feasible for all values of $T_1$ and for all possible degree distributions. Assume, $B=1$, the problem  becomes infeasible if $1 - \psi \leq \gamma$ when $T_2$ is at the maximum possible value, i.e., all individuals are incentivized and yet $1 - \psi \leq \gamma$. 
   
   \subsection{Epidemic Size Maximization Problem}
   We now look at the problem of maximizing the information epidemic size (outreach) in a resource constrained scenario. More, specifically we study a scenario where the cost budget is finite. Thus the outbreak size $1-\psi$ must be maximized subject to a cost constraint. Since $\frac{d\psi}{dq} < 0$,  maximizing $q$ is equivalent to maximizing $1-\psi$. Thus the problem is equivalent to the following linear program. 
   \begin{align}
      \begin{aligned}
      & \underset{\boldsymbol{\phi}}{\text{maximize}}
       \ \ \ \ \sum_{k=1}^{\infty}k\phi(k)P(k) \\
      & \text{subject to:} \ \  \\
      &		  \sum_{k=1}^{\infty} c(k)\phi(k)P(k) \leq C \\
      &          \sum_{k=1}^{\infty} \phi(k)P(k) \leq B \\
      &		\boldsymbol{0} \leq \boldsymbol{\phi} \leq \boldsymbol{1}  \label{eqn:linOpt2}
        \end{aligned}
   \end{align}
   The linear program can now be solved using any standard linear programing solver. Note that constants $T_1, T_2$ do not play any role in problem (\ref{eqn:linOpt2}), while they do play a role in problem (\ref{eqn:LinOpt1}) because $q^*$ is a function of $T_1$ and $T_2$. 
   \section{Numerical Results \label{sec:Numerical Results}}
   As an illustration, we study the solution of the optimization problem for a linear cost, i.e., $c(k)=k$. The higher the degree, the higher the cost. Note that even if cost is non linear in $k$, the optimization problem remains a linear program. In the real world, the cost may be different, but whatever the cost function, the solution can be obtained by simply solving a linear program.
   \par
   We used an uncorrelated random graph generated using the configuration model technique with power law degree distribution ($P(k)\propto k^{- \alpha}$), $\alpha = 2.5$.
   \subsection{Cost Minimization Problem}
      \begin{figure}
      \centering
      
     \begin{subfigure}[b]{0.49\textwidth}
        \includegraphics[width = \textwidth]{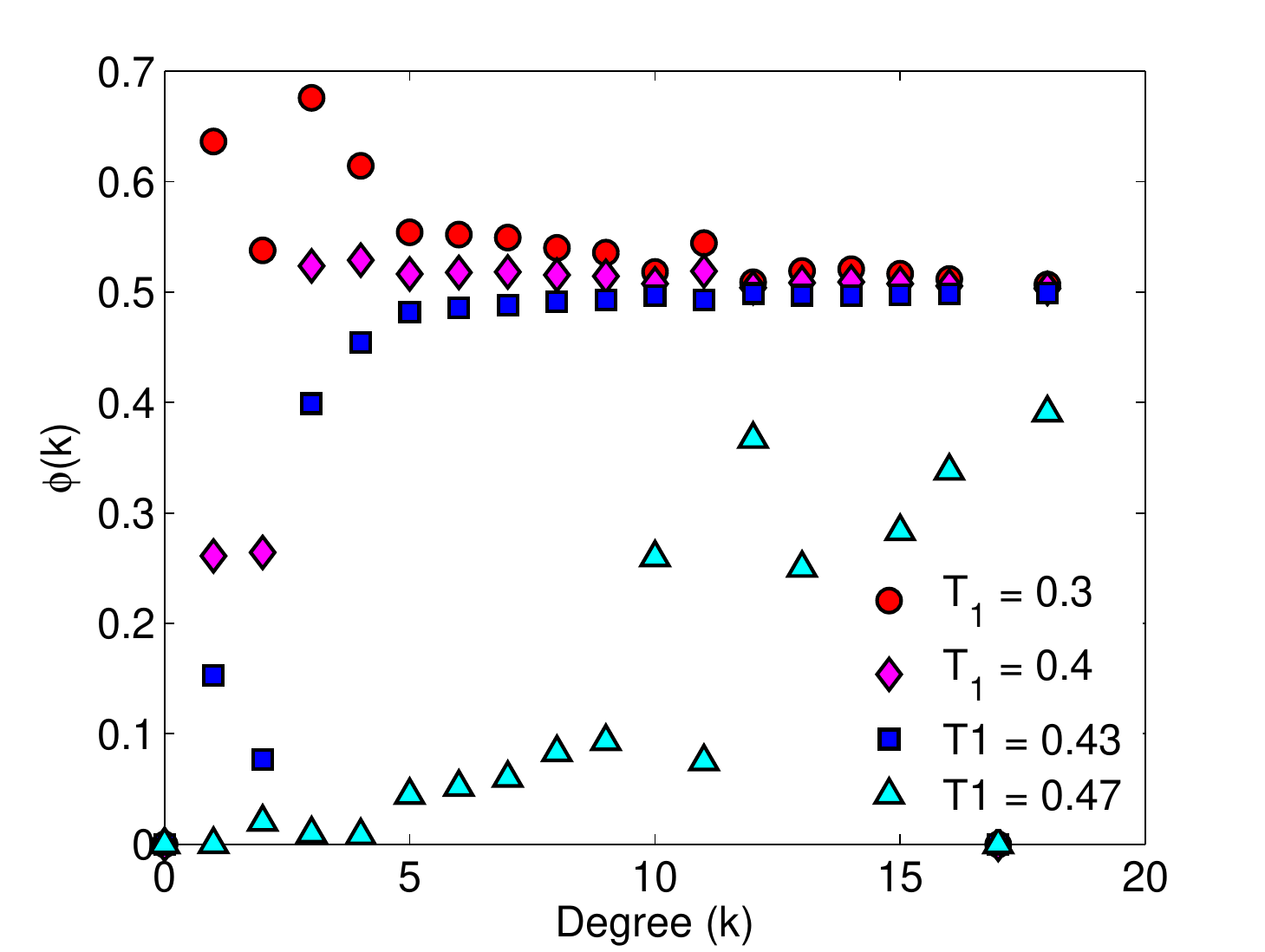}
       
                      \caption{}
       \label{fig:CostMinSoln}
       \end{subfigure}
              \begin{subfigure}[b]{0.49\textwidth}
               \includegraphics[width = \textwidth]{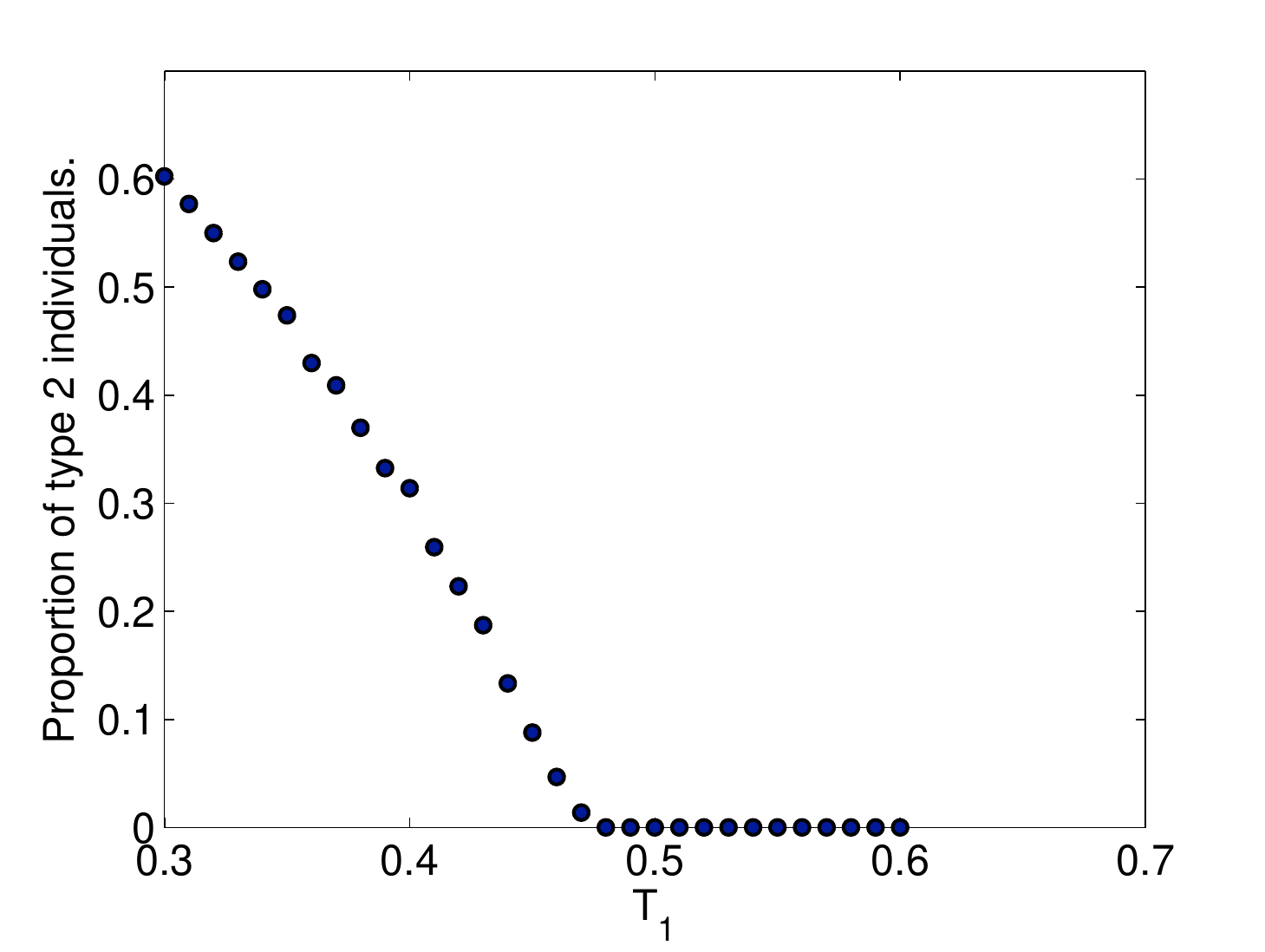}
               \caption{}           
              \label{fig:Type2ForT1}
             \end{subfigure}
         \caption{(Color Online) (a) Solution $\boldsymbol{\phi}$, for different values of $T_1$; (b) Optimal proportion of type 2 nodes required to meet outreach constraint. Parameters: $T_2 = 0.6,\ B = 0.7, \ \gamma = 0.2$.}
       \end{figure}
       
       We solved the cost minimization linear program using the `\emph{linprog}' MATLAB solver; $q^*$ was computed numerically using the bisection method. In Fig. \ref{fig:CostMinSoln}, we plot the solution $\boldsymbol{\phi}$  for different values of $T_1$.  The solution shows that only about $50 \%$ of high degree nodes need to be incetivized for $T_1$ values ranging from $0.3$ to $0.43$. As $T_1$ decreases from $0.47$ to $0.3$ , the proportion of high degree nodes that are incentivized remain fairly constant (50\%), while the proportion of incentivized low degree nodes increase. In Fig. \ref{fig:Type2ForT1}, we plot the optimal proportion of individuals that need to be incentivized for achieving the given outreach size.   
   \subsection{Epidemic Size Maximization Problem}
        \begin{figure}
         \centering
         
        \begin{subfigure}[b]{0.49\textwidth}
           \includegraphics[width = \textwidth]{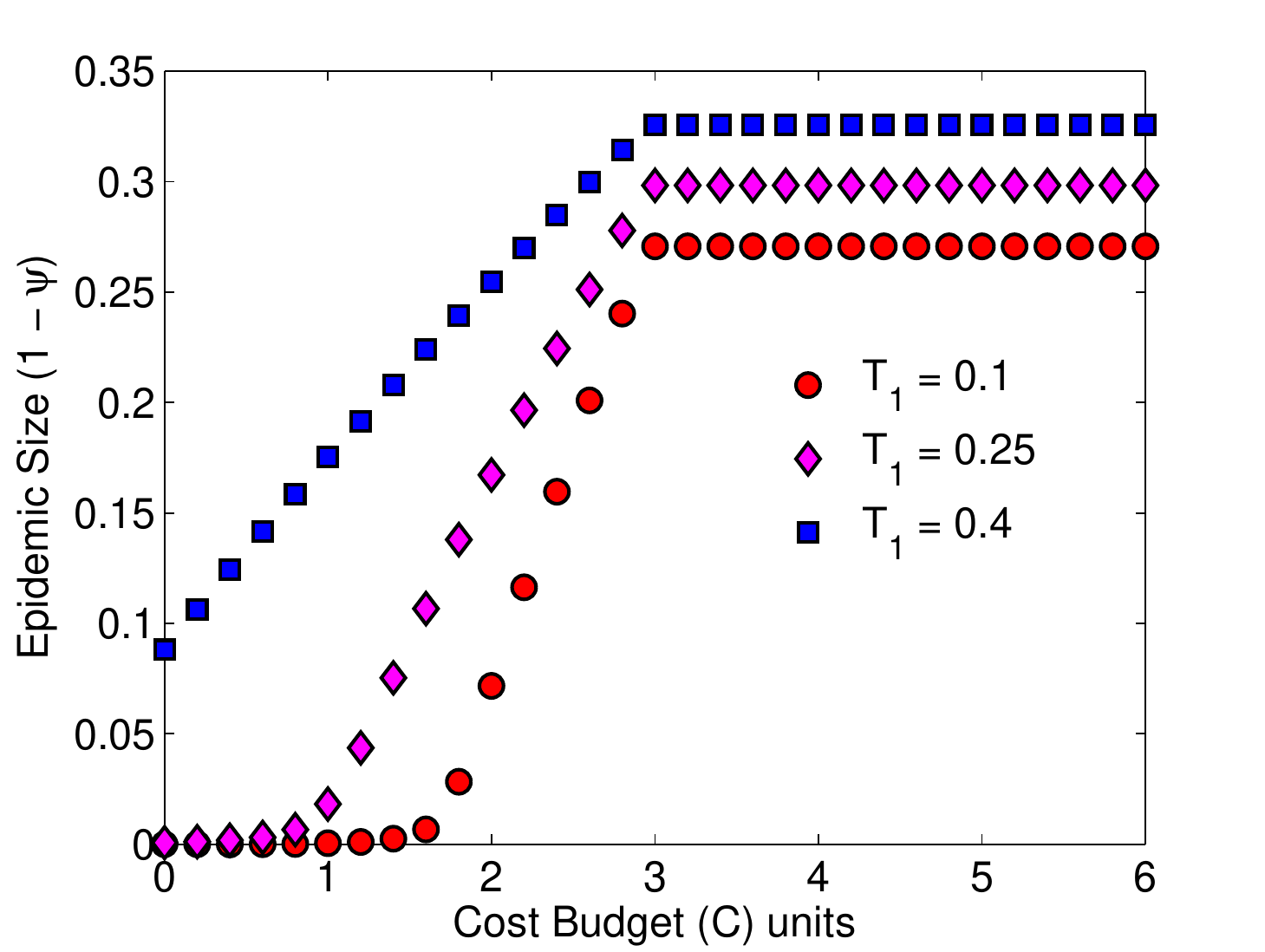}
          
                         \caption{}
          \label{fig:SizeVsCost}
          \end{subfigure}
                 \begin{subfigure}[b]{0.49\textwidth}
                  \includegraphics[width = \textwidth]{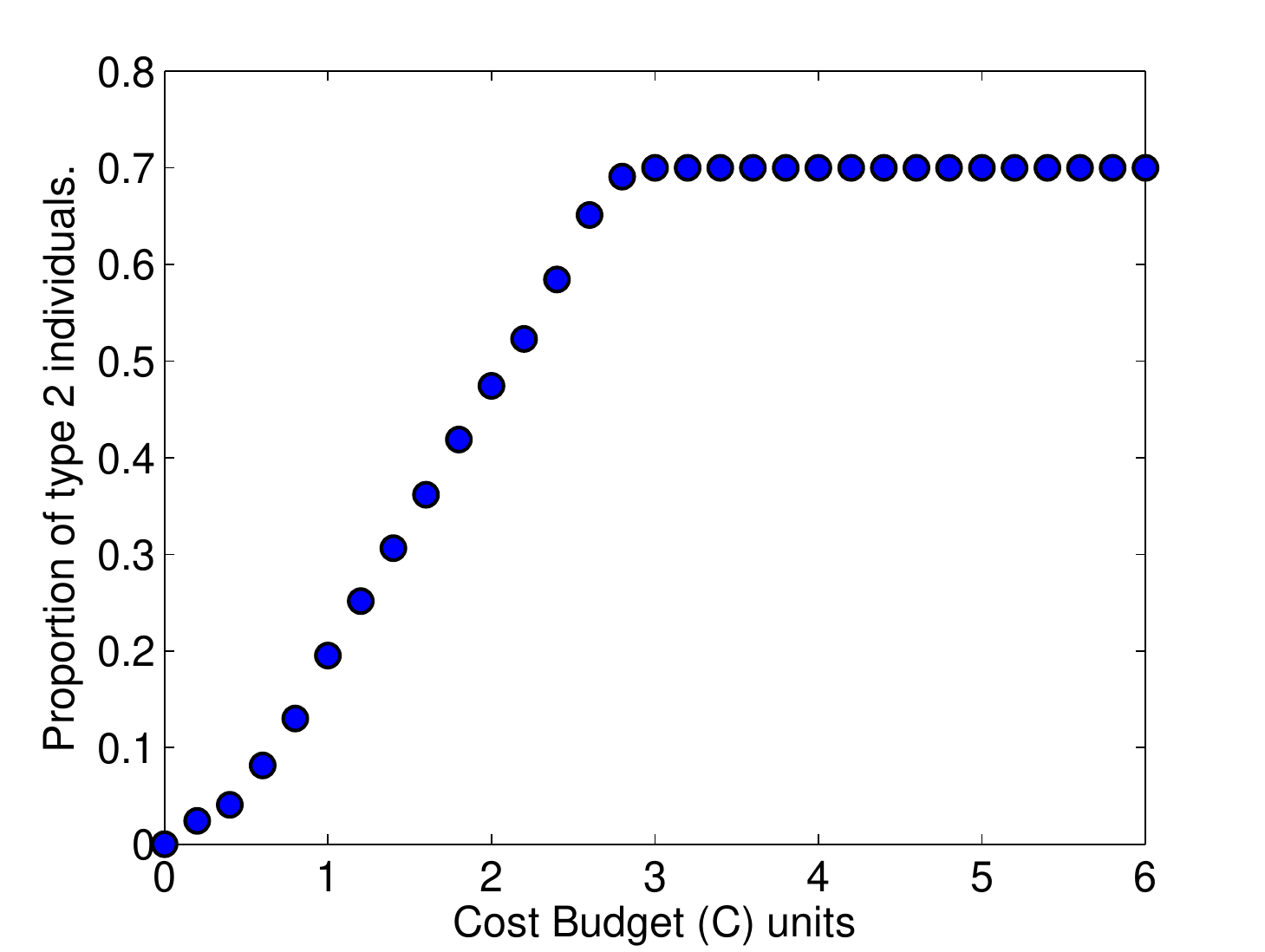}
                  \caption{}           
                 \label{fig:Type2ForCost}
                \end{subfigure}
            \caption{(Color Online) (a) Size of the information epidemic for varying cost budget (C). (b) Optimal proportion of type 2 nodes required to meet outreach constraint as a function of cost budget (C). Parameters: $T_2 = 0.6,\ B = 0.7$.}
          \end{figure}
          The solution,  $\boldsymbol{\phi}$, is very similar to the one in problem (\ref{eqn:LinOpt1}), and hence we do not show it here. In Fig.  \ref{fig:SizeVsCost}, we plot the size of the epidemic for varying cost budget $C$. As expected, the epidemic size increases with $C$ because higher the budget, the higher the proportion of incentivized individuals. However, at some point epidemic size saturates, this is because all nodes have been incentivized and therefore nothing more can be done to increase the outreach size. This is verified by Fig. \ref{fig:Type2ForCost}, the fraction of type 2 nodes hit $1$, when $C=3$.   
          
          \section{Conclusion and Future Work\label{sec:Conclusion}}
    To summarize, we studied the problem of maximizing information spreading in a social networks.  More specifically, we considered a scenario where individuals are incentivized to vigorously spread the campaign message to their neighbors, and we proposed a mechanism to identify the individuals who should be incetivized. Using bond percolation theory we calculated the size of the information epidemic outbreak and the conditions for the occurrence of such outbreaks. We then formulated an optimization problem for minimizing the expected cost of incentivizing individuals while providing guarantees on the information epidemic size. Although the optimization problem could not be addressed using standard analytical tools, Theorem \ref{theorem2} enabled us to compute the global optimum by solving a linear program. We believe that our approach of using percolation theory to formulate an optimization problem is the first of its kind.
                 \par
   { For the sake of analytical tractability we assumed an uncorrelated network, however in reality real world social networks have positive degree-degree correlations \cite{Newman2003b}. Such networks with positive degree associativity percolate more easily compared to uncorrelated networks \cite{Newman2002b, Noh2007}. Therefore, for the problem of minimizing cost the given campaign size could be achieved with a slightly lesser cost, while in the second problem, the theoretical optimal size would be a lower bound and the actual campaign size would be slightly larger than the theoretical. Apart from positive degree associativity, social networks are also found to contain community structures \cite{Newman2003b}. The presence of communities may slow down information spreading leading to a reduction in the campaign size. This may happen as most links point inside the community rather than outside it, thus localizing the information spread \cite{Wu2008}. However, if the network contains high degree nodes that bridge different communities, then incetivizing such nodes may substantially increase the campaign size. A similar finding was reported in \cite{Salathe2010}, where authors investigated usefulness of  targeted vaccinations on nodes that bridge communities.}
   \par
{ Although SIR models are widely used to model epidemics, they have some limitations. They fail to capture the fact that individuals may stop spreading when they perceive that most of their neighbors already known the information. This is captured by the  Maki-Thompson model \cite{Nekovee2007} which forces the recovery rate to be an increasing function of the number of informed individuals she contacts. Thus the recovery rate for an infected node is a function of her degree. An SIR process has a fixed recovery rate and hence the current results would approximately hold for an Maki Thompson process on Erdos-Renyi networks, where every node on average has the same degree. However, our results for SIR may not generalize for the Maki Thompson spread model on scale free networks. High degree nodes may have a higher chance of being connected to informed individuals which may lead them to stop spreading to other uninformed nodes. }
 \par
 { An interesting extension to this problem, which was suggested by the anonymous referee, is to compute a targeted incentivization strategy for two interacting campaigns. For example, the campaigner may want to maximize campaign $A$ given that campaign $B$, which has either run its course or is simultaneously running along with $A$, either reinforces or hinders campaign $A$. This is an important problem since such interacting campaigns are often observed during parliamentary or presidential elections. Although the current results may not shed much light on such questions, we believe that they lay the foundation for investigating such problems which we hope to address in the future. }
 \appendix
\section{Proof of Theorem 3.1 \label{appendix:theorem1}}
\begin{proof}
Let $\langle s_1 \rangle$ and $\langle s_2 \rangle$ be the average number of type 1 and type 2 nodes in the component. The expected number of nodes in the component, $<s>$, is given by:
\begin{align*}
\langle s \rangle &= \langle s_1 \rangle + \langle s_2 \rangle  \\
&= \frac{\partial}{\partial u_1}\tilde{J}(u_1,u_2) \biggr|_{\boldsymbol{u}=1} + \frac{\partial}{\partial u_2}\tilde{J}(u_1,u_2)  \biggr|_{\boldsymbol{u}=1}
\end{align*}
After differentiating and simplifying, $\langle s_1 \rangle$ can be written as:
\begin{align*}
\langle s_1 \rangle = (1-p) + \langle \tilde{k}_1 \rangle \tilde{H}_1^{'}(1,1) + \langle \tilde{k}_2 \rangle \tilde{H}_2^{'}(1,1) 
\end{align*}
where $\langle \tilde{k}_i \rangle  = \sum \limits_{\tilde{k}_1,\tilde{k}_2}^{\infty}\tilde{k}_i \tilde{P}(\tilde{k}_1,\tilde{k}_2)$  and  
\begin{align*}
\tilde{H}_i^{'}(1,1) = \frac{\partial}{\partial u_1} \tilde{H}_i(u_1,u_2) \biggr |_{\boldsymbol{u} = 1}
\end{align*}
$\tilde{H}_i^{'}(1,1)$ can be obtained by differentiating equation (\ref{eqn:pgfCluster}).
\begin{align*}
\tilde{H}_1^{'}(1,1)&= 1 + T_1 \bar{k}_1\tilde{H}_1^{'}(1,1) + T_2 \bar{k}_2\tilde{H}_2^{'}(1,1) \\
\tilde{H}_2^{'}(1,1)&= T_1 \bar{k}_1\tilde{H}_1^{'}(1,1) + T_2 \bar{k}_2\tilde{H}_2^{'}(1,1)
\end{align*}
where $\bar{k}_i = \sum \limits_{k_1,k_2}^{\infty}k_iQ(k_1,k_2)$. Solving the two simultaneous equations we obtain $\tilde{H}_1^{'}(1,1)= \frac{1-T_2\bar{k}_2}{1 - T_1\bar{k}_1 - T_2\bar{k}_2} $ and $\tilde{H}_2^{'}(1,1)= \frac{T_1\bar{k}_1}{1 - T_1\bar{k}_1 - T_2\bar{k}_2}$. Substituting in the expression for $\langle s_1 \rangle$ we get.
\begin{align*}
\langle s_1 \rangle = (1-p) + \frac{\langle\tilde{k}_1\rangle(1-T_2\bar{k}_2) + \langle\tilde{k}_2\rangle T_1\bar{k}_1 }{1 - T_1\bar{k}_1 - T_2\bar{k}_2}
\end{align*}
One can similarly show that:
\begin{align*}
\langle s_2 \rangle =  p + \frac{\langle \tilde{k}_1 \rangle T_2\bar{k}_2 + \langle \tilde{k}_2 \rangle(1-T_1\bar{k}_1) }{1 - T_1\bar{k}_1 - T_2\bar{k}_2}
\end{align*}
Therefore,
\begin{align*}
\langle s \rangle =  1 + \frac{\langle\tilde{k}_1\rangle + \langle \tilde{k}_2\rangle}{1 - T_1\bar{k}_1 -T_2 \bar{k}_2}
\end{align*}
Thus, when $T_1\bar{k}_1 + T_2\bar{k}_2 \geq 1$, $\langle s \rangle$ is no longer finite, it morphs into a giant connected component, or in other words there is an information epidemic outbreak. 
\par
Assume that a giant connected component of exists ($\tilde{\nu} \geq 1$). For any given node let $z_i$ be the probability that one of its type $i$ links does \emph{not} lead to the giant connected component. The probability that a randomly chosen node is not a part of the GCC  is given by
\begin{align*}
\psi &= \sum_{\tilde{k}_1,\tilde{k}_2}^{\infty}z_1^{k_1}z_2^{k_2}P(k_1,k_2) \\
\end{align*} 
Now, $z_i$ can be written as $Pr($link is not occupied $)$ + $Pr($link is occupied and the neighbor is not connected to the GCC$)$. By occupied we mean that the node at the end of the link is a believer. Mathematically this can be written as: 
\begin{align*}
z_1 &=  1 - T_1 + T_1\sum_{k_1,k_2}^{\infty}z_1^{k_1}z_2^{k_2}Q(k_1,k_2)  \\
z_2 &=  1 - T_2 + T_2\sum_{k_1,k_2}^{\infty}z_1^{k_1}z_2^{k_2}Q(k_1,k_2)
\end{align*}
Simplifying we obtain, $(z_1 -1)/T_1 = (z_2 -1)/T_2 $. Let $u := (z_1 -1)/T_1 + 1$. Hence, $z_i = 1 + (u-1)T_i$. Note that  $z_i$ is bounded from below by $1-T_i$ and bounded from above by $1$, and hence $0 \leq u \leq 1$. Substituting this in above equations we obtain the desired result:
\begin{align*}
\psi = \sum_{k_1,k_2}^{\infty}(1+(u-1)T_1)^{k_1}(1+(u-1)T_2)^{k_2}P(k_1,k_2)
\end{align*}
where $u$ must satisfy
\begin{align*}
u = \sum_{k_1,k_2}^{\infty}(1+(u-1)T_1)^{k_1}(1+(u-1)T_2)^{k_2}Q(k_1,k_2)
\end{align*}

\end{proof}

\section{Lemmas required for Theorem 4.1 \label{appendix:theorem2}}
Lemma \ref{lemma1} is used in the proof of Lemma \ref{lemma2} and \ref{lemma3}.
 \begin{lem} \label{lemma1}
 For all $a ,\ b \ \in \ [0,1] $ and $ k_1 + k_2 \leq n, \ n \ \in \ \boldsymbol{Z}^+$ and any arbitrary $f:\boldsymbol{Z}\rightarrow\boldsymbol{R}$ the following is true:
  \begin{align*}
 &\sum_{k_1=0}^{n}\sum_{k_2=0}^{n-k_1}f(k_1+k_2)k_2 {k_1+ k_2 \choose k_2}a^{k_2-1}b^{k_1} \\ -& \sum_{k_1=0}^{n}\sum_{k_2=0}^{n-k_1}f(k_1+k_2)k_1 {k_1 + k_2 \choose k_2}a^{k_2}b^{k_1-1} = 0
 \end{align*}
  \end{lem}
 \begin{proof}
 We can switch the indices in the second term, i.e.,
 \begin{align*}
 &\sum_{k_1=0}^{n}\sum_{k_2=0}^{n-k_1}f(k_1+k_2)k_1 {k_1 + k_2 \choose k_2}a^{k_2}b^{k_1-1}  \\
 &=\sum_{k_1=0}^{n}\sum_{k_2=0}^{n-k_1}f(k_1+k_2)k_2 {k_1 + k_2 \choose k_2}a^{k_1}b^{k_2-1}
 \end{align*}
 Hence,
  \begin{align}
  LHS = &\sum_{k_1=0}^{n}\sum_{k_2=0}^{n-k_1}  f(k_1+k_2)k_2 {k_1+ k_2 \choose k_2}\bigg{(}a^{k_2-1}b^{k_1} - a^{k_1}b^{k_2-1} \bigg{)} \nonumber \\ 
  & = \sum_{k_1=0}^{n}\sum_{k_2=0}^{n-k_1} g(k_1,k_2) \label{eqn:lemma1} 
  \end{align}
 We now count the number of terms in the above equation and show that they are even.  An expression indexed by a specific $k_1$ and $k_2$ denotes a term, e.g, $g(1,1)$ is a term. The total number of terms in the summation $= \sum\limits_{i=1}^{n+1}i = \frac{(n+1)(n+2)}{2}$. Out of those,  $n+1$ terms are $0$ due to the $k_2$ multiplier ($k_2 = 0$ for $k_1 = 0 \ \text{to} \ n$). Additionally, when $k_2 = k_1+1$  equation (\ref{eqn:lemma1}) is zero. The total number of terms when $k_2 = k_1 +1 $ is given by $\floor*{\frac{n+1}{2}}$.
 \par
 Since, these terms are zero, subtracting out these terms from the total number of terms results in 
 \begin{align*}
 &\frac{(n+1)(n+2)}{2} - (n+1) - \floor*{\frac{n+1}{2}} \\
 &= \frac{n^2}{2} \ \text{for }n \text{ even}\\
 & = \frac{(n-1)(n+1)}{2} \ \text{for }n \text{ odd}
 \end{align*}
 
Thus, the remaining terms are even for both $n$ odd and even. This allows us to pair the terms.  Consider one such pairing: the term with indices $k_1, \ k_2$ are paired with a term with indices $\hat{k}_1, \ \hat{k}_2$ where $\hat{k}_2 = k_1 + 1$ and $\hat{k}_1 = k_2 -1$.  If we sum these two terms we obtain

 \begin{align*}
 & g(k_1,k_2) + g(\hat{k}_1,\hat{k}_2) \\
 & = f(k_1+k_2)a^{k_2-1}b^{k_1}\left(  \frac{k_2(k_1+k_2)!}{k_1!k_2!} - \frac{k_2(k_1+k_2)!}{k_1!k_2!}  \right)\\&+  f(k_1+k_2)a^{k_1}b^{k_2-1}\left( \frac{(k_1+1)(k_1+k_2)!}{(k_2-1)!(k_1+1)!} - \frac{(k_1+1)(k_1+k_2)!}{(k_2-1)!(k_1+1)!} \right)  \\
 &=0
 \end{align*}
 
  Thus, the summation of the remaining terms is zero, which completes the proof. 
  \end{proof}

    \begin{lem}\label{lemma2}
    If $T_2 >T_1$ then $\tilde{\nu}$ is strictly increasing with respect to $q$, i.e, $\frac{d\tilde{\nu}}{dq} > 0, \ \forall \ q \ \in \ [0,1]$.
    \end{lem}
    \begin{proof}
   
    \begin{align*}
  \frac{d\tilde{\nu}}{dq} =  & T_1\sum_{k_1,k_2}^{\infty}k_1Q(k_1+k_2){k_1+k_2 \choose k_2}\left(k_2q^{k_2-1}r^{k_1} - k_1q^{k_2}r^{k_1-1}\right) \\
    &+ T_2\sum_{k_1,k_2}^{\infty}k_2Q(k_1+k_2){k_1+k_2 \choose k_2}\left(k_2q^{k_2-1}r^{k_1} - k_1q^{k_2}r^{k_1-1}\right)
    \end{align*}
    
    where $r = 1-q$. Let,

    \begin{align*}
    a_1 &= \sum_{k_1,k_2}^{\infty}k_1Q(k_1+k_2){k_1+k_2 \choose k_2}\left(k_2q^{k_2-1}r^{k_1} - k_1q^{k_2}r^{k_1-1}\right) \\ 
    a_2 &= \sum_{k_1,k_2}^{\infty}k_2Q(k_1+k_2){k_1+k_2 \choose k_2}\left(k_2q^{k_2-1}r^{k_1} - k_1q^{k_2}r^{k_1-1}\right)
    \end{align*}
    
    Adding $a_1$ and $a_2$ we get

    \begin{align*}
     a_1 + a_2 =  \sum_{k_1,k_2}^{\infty}(k_1+k_2)Q(k_1+k_2){k_1+k_2 \choose k_2}\left(k_2q^{k_2-1}r^{k_1} - k_1q^{k_2}r^{k_1-1}\right)
    \end{align*}
    
     Since $N$ is large but finite, $Q(k_1+k_2)=0$ for $k_1+k_2 > k_{max}$, where $k_{max}$ is the maximum degree. From Lemma \ref{lemma1}, $a_1 + a_2 = 0$. Now we prove that $a_2 > 0$. Let $k_1 + k_2 = m$. 
      \begin{align*}
    a_2 = \sum_{m=1}^{k_{max}}Q(m)\bigg[ \frac{1}{q}\sum_{k_2=0}^{m}k_2^2{m \choose k_2}q^{k_2}r^{m-k_2}  -\frac{1}{r}\sum_{k_1=0}^{m}k_1(m-k_1){m \choose k_1}q^{m-k_1}r^{k_1}\bigg ]
    \end{align*}
    The summations are the second moments of a binomial random variable. $E \left[X^2 \right] = Var\left[X\right] + E \left[X \right]^2$, $E[X] = mq, \ Var[X] = mqr$.
   
    \begin{align*}
    a_2 &= \sum_{m=1}^{k_{max}}Q(m)\bigg[\frac{1}{q}(mqr + m^2q^2)  -\frac{1}{r}(m^2r - mqr-m^2r^2)\bigg ] \\
    &= \sum_{m=1}^{k_{max}}Q(m)m \\
    &>0
    \end{align*}
    
      Since $T_2>T_1$, $T_1a_1 + T_2a_2 > 0$, which completes the proof.
    \end{proof}
    
    \begin{lem} \label{lemma3}
    For $\psi \ \in \  (0,1)$, if $T_2 >T_1$ then $\psi$ is strictly decreasing with respect to $q$, i.e, $\frac{d\psi}{dq} <0,  \ \forall \ q \ \in \ [0,1]$. 
    \end{lem}
    \begin{proof}
    Let, $\psi = g(u^*,q)$ where $u^*$ is the solution of the fixed point equation $u = f(u,q)$.
    
    \begin{align*}
    g(u^*,q) &= \sum_{k_1,k_2}^{\infty}\alpha^{k_1}\beta^{k_2}P(k_1+k_2)  {k_1+k_2 \choose k_2}q^{k2}(1-q)^{k_1} \\
    f(u,q)&=\sum_{k_1,k_2}^{\infty}\alpha^{k_1}\beta^{k_2}Q(k_1+k_2){k_1+k_2 \choose k_2}q^{k2}(1-q)^{k_1}
    \end{align*}
         where $\alpha = 1+(u^*-1)T_1$ and $\beta = 1+(u^*-1)T_2$. We first show that the solution to the fixed point equation is strictly decreasing with $q$ . 
       \par
    Let us consider the behavior of the R.H.S  of the fixed point equation, $f(u,q)$, w.r.t. $q$.  Now  
    \begin{align*}
     \frac{\partial f(u,q)}{\partial q} =&\sum_{k_1,k_2}^{\infty}\alpha^{k_1}\beta^{k_2}Q(k_1+k_2){k_1+k_2 \choose k_2}\left(k_2q^{k2-1}r^{k_1}-k_1q^{k2}r^{k_1-1}\right) \\
     &=\beta\sum_{k_1,k_2}^{\infty}Q(k_1+k_2){k_1+k_2 \choose k_2}k_2(\beta q)^{k2-1}(\alpha r)^{k_1} \\
     &-\alpha\sum_{k_1,k_2}^{\infty}Q(k_1+k_2){k_1+k_2 \choose k_2}k_1(\beta q)^{k_2}(\alpha r)^{k_1-1} 
    \end{align*}
    
      From Lemma \ref{lemma1},  
    \begin{align*}
    &\sum_{k_1,k_2}^{\infty}Q(k_1+k_2){k_1+k_2 \choose k_2}k_2(\beta q)^{k2-1}(\alpha r)^{k_1} \\
     &-\sum_{k_1,k_2}^{\infty}Q(k_1+k_2){k_1+k_2 \choose k_2}k_1(\beta q)^{k_2}(\alpha r)^{k_1-1}  \\
     & = 0
    \end{align*}
        Now $\alpha > \beta$ because $T_2 > T_1$, which implies  $ \frac{\partial f(u,q)}{\partial q} <0$.
      \par
    We use the \emph{Implicit Function Theorem} for computing the sign of $\frac{\partial u^*}{\partial q}$.  Let $h(u,q) = f(u,q) - u$. According to the Implicit Function Theorem when $h(u,q) = 0$ 
    \begin{align*}
    \frac{du}{dq} = -\frac{\frac{\partial }{\partial q}h(u,q)}{\frac{\partial }{\partial u}h(u,q)}
    \end{align*}
   We now show that the derivative exists and is greater than zero.  The numerator  $\frac{\partial h(u,q)}{\partial q} = \frac{\partial f(u,q)}{\partial q} < 0 $.
        The denominator is given by 
   $  \frac{\partial h(u,q)}{\partial u} =  \frac{\partial f(u,q)}{\partial u} - 1 $
     \par
        Claim:  $\frac{\partial f(u,q)}{\partial u} < 1$ when $h(u,q)=0$.  We prove this by contradiction. Assume the contrary, i.e.,  $\frac{\partial f(u,q)}{\partial u} \geq 1$. 
     \par
      It can be easily shown that $\frac{\partial f(u,q)}{\partial u} > 0 $ and $\frac{\partial^2 f(u,q)}{\partial u^2} > 0$ for all $u, \ q \ \in [0,1]$. Thus $f$ is a convex function in $u$ for any fixed $q$. Also $f(0,q) > 0$ for all $q \ \in [0,1]$. 
      \par
      Now $h(u,q) = 0$ implies that $u=f(u,q)$, or in other words the curve $f(u,q)$ intersects the line passing through the origin with slope $1$. Since we assumed $\frac{\partial f(u,q)}{\partial u} \geq 1$, i.e., the derivative of $f(u,q)$ is greater than equal to $1$ at the intersection, and since $f(u,q)$ is monotone increasing in $u$, the curve $f(u,q)$ will never again intersect the line passing through the origin with slope $1$. Therefore the equation $u=f(u,q)$ has a unique fixed point. However, this is a contradiction since $u=1$ is always a fixed point and since $\psi > 0$ there is another fixed point less than $1$. Hence,   $\frac{\partial f(u,q)}{\partial u} \geq 1$ is impossible, and therefore $\frac{\partial f(u,q)}{\partial u} < 1$.
      \par
      Thus, the derivative exists and is less than $0$ as  $\frac{\partial h(u,q)}{\partial q} <0$ and $ \frac{\partial }{\partial u}h(u,q)<0$. Since we assumed $h(u,q)=0$, the derivative $\frac{du}{dq}$ is the derivative of the fixed point w.r.t $q$, i.e., it can be represented as $\frac{du^*}{dq}$, where $u^*$ is the fixed point 
    The function $g(u^*,q)$ has the same structure as the function $f(u,q)$, and hence using the same procedure it can be shown that $\frac{\partial g(u,q)}{\partial q} < 0$. The total derivative  $\frac{d\psi}{dq}$ is given by:
    \begin{align*}
    \frac{d\psi}{dq} = \frac{\partial g}{\partial q} + \frac{\partial g}{\partial u}\frac{du^*}{dq}
    \end{align*}   
    Since all the terms on the right hand side of the above equation are negative ($g$ is a non decreasing function of u), $\frac{d\psi}{dq} < 0$. 
    \end{proof}




\bibliographystyle{model1-num-names}
\bibliography{Information_Rumor_Spread}







\end{document}